\documentclass[11pt]{article}
\usepackage{amsmath,amssymb,amsbsy,amsfonts,amsthm,latexsym,
               amsopn,amstext,amsxtra,euscript,amscd,color}
\def \F {{\mathbb F}}
\def \Q {{\mathbb Q}}
\def \Z {{\mathbb Z}}

\def \V {{\mathbb V}}

\newtheorem{theorem}{Theorem}

\newtheorem{lemma}[theorem]{Lemma}
\newtheorem{proposition}[theorem]{Proposition}
\newtheorem{remark}[theorem]{Remark}
\newtheorem{example}[theorem]{Example}

\def\cB{{\mathcal B}}
\def\cC{{\mathcal C}}

\def\cH{{\mathcal H}}

\def\cW{{\mathcal W}}

\def\cGB{\mathcal{GB}}

\def\aa{{\bf a}}

\def\uu{{\bf u}}
\def\vv{{\bf v}}

\def\xx{{\bf x}}

\def\zz{{\bf z}}

\def\00{{\bf 0}}
\def\11{{\bf 1}}
\def\+{\oplus}

\def \F {{\mathbb F}}
\def \Q {{\mathbb Q}}
\def \Z {{\mathbb Z}}

\def \V {{\mathbb V}}

\newcommand{\BBF}{\mathbb{F}}

\begin{document}

\title{\huge\bf
\textrm{Generalized bent functions and their Gray images} }

\author{\Large  Thor Martinsen$^1$, Wilfried Meidl$^2$, \and \Large  Pantelimon St\u anic\u a$^1$
\vspace{0.4cm} \\
\small $^1$Department of Applied Mathematics, \\
\small Naval Postgraduate School, Monterey, CA 93943-5212, U.S.A.;\\
\small Email: {\tt \{tmartins,pstanica\}@nps.edu}\\
\small $^2$Johann Radon Institute for Computational and Applied Mathematics,\\
\small Austrian Academy of Sciences, Altenbergerstrasse 69, 4040-Linz, Austria;\\
\small Email: {\tt meidlwilfried@gmail.com}
}

\date{\today}
\maketitle
\thispagestyle{empty}

\begin{abstract}
In this paper we prove that generalized bent (gbent) functions defined on $\mathbb{Z}_2^n$ with values in $\mathbb{Z}_{2^k}$ are regular, and find connections
between the (generalized) Walsh spectrum of these functions and their components. We comprehensively characterize generalized bent and semibent functions with
values in $\mathbb{Z}_{16}$, which extends earlier results on gbent functions with values in $\mathbb{Z}_4$ and $\mathbb{Z}_8$. We also show that the Gray
images of gbent functions with values in $\mathbb{Z}_{2^k}$ are semibent/plateaued when $k=3,4$.
\end{abstract}

\section{Introduction}

Let $\V_n$ be an $n$-dimensional vector space over $\F_2$ and for an integer $q$, let $\Z_q$ be the ring of integers modulo $q$.
We label the real and imaginary parts of a complex number $z=\alpha+\beta i$, $\alpha,\beta\in{\mathbb R}$, by $\Re(z)=\alpha$ and $\Im(z)=\beta$, respectively.
For a {\it generalized Boolean function} $f$ from $\V_n$ to $\Z_q$ the {\it generalized Walsh-Hadamard transform} is the complex valued function
\[ \mathcal{H}^{(q)}_f(\uu) = \sum_{\xx\in \V_n}\zeta_q^{f(\xx)}(-1)^{\langle\uu,\xx\rangle},\quad  \zeta_q = e^{\frac{2\pi i}{q}}, \]
where $\langle\uu,\xx\rangle$ denotes a (nondegenerate) inner product on $\V_n$ (we shall use $\zeta$,  $\cH_f$, instead of $\zeta_q$, respectively, $\cH_f^{(q)}$,  when $q$ is fixed).
In this article we always identify $V_n$ with the vector space $\F_2^n$ of $n$-tuples over $\F_2$, and we use the regular scalar (inner) product $\langle\uu,\xx\rangle = \uu\cdot\xx$.
We denote the set of all generalized Boolean functions by $\mathcal{GB}_n^q$ and when $q=2$, by $\mathcal{B}_n$.
A function $f:\V_n\rightarrow\Z_q$ is called {\em generalized bent} ({\em gbent}) if $|\mathcal{H}_f(\uu)| = 2^{n/2}$ for all $\uu\in \V_n$.

We recall that for $q=2$, where the generalized Walsh-Hadamard transform of $f$ reduces to the conventional {\it Walsh-Hadamard transform}
\[ \mathcal{W}_f(\uu) = \sum_{\xx\in \V_n}(-1)^{f(\xx)}(-1)^{\uu \cdot \xx }, \]
a function $f$ for which $|\mathcal{H}_f(\uu)| = 2^{n/2}$ for all $\uu\in \V_n$ is called a {\em bent} function. Further recall that $f\in\mathcal{B}_n$ is called {\em plateaued} if
$|\mathcal{W}_f(\uu)| \in \{0,2^{(n+s)/2}\}$ for all $\uu\in \V_n$ for a fixed integer $s$ depending on $f$ (we also call $f$ then $s$-{\em plateaued}). If $s=1$ ($n$ must then be odd),
or $s=2$ ($n$ must then be even), we call $f$ semibent. With this notation a semibent function is an $s$-plateaued Boolean function with smallest possible $s>0$.
Accordingly we call a function $f\in \mathcal{GB}_n^q$, with $q=2^k$, $k>1$ (the case in which we will be most interested), {\it generalized plateaued} if
$|\mathcal{H}_f(\uu)| \in \{0,2^{(n+s)/2}\}$ for all $\uu\in \V_n$ and some integer $s$, and {\em generalized semibent} ({\em gsemibent}, for short) if
$|\mathcal{H}_f(\uu)| \in \{0,2^{(n+1)/2}\}$ for all $\uu\in \V_n$. Note that differently to a Boolean function, where $k=1$, a generalized Boolean function
$f\in\mathcal{GB}_n^{2^k}$, $k>1$, can be generalized $1$-plateaued also if $n$ is even.

If $f$ is gbent such that for every $\uu\in \V_n$, we have $\mathcal{H}_f(\uu) = 2^{n/2}\zeta_q^{f^*(\uu)}$ for some function $f^*\in\mathcal{GB}_n^q$,
then we call $f$ a {\em regular} gbent function. Similar as for bent
functions we call $f^*$ 
the {\em dual} of $f$. With the same argument as for the conventional bent functions we can see that the dual $f^*$ is also gbent and $(f^*)^* = f$.

The sum $$\cC_{f,g}(\zz)=\sum_{\xx \in \V_n} \zeta ^{f(\xx)  - g(\xx \+ \zz)}$$
is  the {\em crosscorrelation} of $f$ and
$g$ at $\zz$.
The {\em autocorrelation} of $f \in \cGB_n^q$ at $\uu \in \V_n$
is $\cC_{f,f}(\uu)$ above, which we denote by $\cC_f(\uu)$.
Two functions $f , g \in \cGB_n^q$ are said to have {\em complementary autocorrelation}
if and only if $\cC_f(\uu) + \cC_g(\uu) = 0$ for all $\uu\in \V_n\setminus \{0\}$.

Let $f:\V_m\to \Z_q$. 
If $2^{k-1}<q\leq 2^k$, we associate a unique sequence of Boolean functions $a_i:\V_m\to\F_2$, $1\leq i\leq k$, such that
\[
f(\xx)=a_1(\xx)+\cdots +2^{k-1} a_k(\xx), \text{ for all } \xx\in\V_m.
\]
If $q=2^k$, following Carlet~\cite{Car98}, we further define the {\em generalized Gray map} $\psi(f): \mathcal{GB}_n^q\to \mathcal{B}_{n+k-1}$ by
\[
\psi(f)(\xx,y_1,\ldots,y_{k-1})=\bigoplus_{i=1}^{k-1} a_i(\xx) y_i \+a_k(\xx).
\]
It is known~\cite{Car98} that the reverse image of the Hamming distance by the generalized Gray map is a translation-invariant distance.

Generalizations of Boolean bent functions, like negabent functions and the more general class of gbent functions have lately attracted increasing
attention, see e.g.~\cite{CS09,gps,HP,pp,KUS1,KUS2,ST09,sgcgm,smgs,spt,Tok} and references therein.


In \cite{ST09,smgs} gbent functions $f(\xx)=a_1(\xx)+2a_2(\xx)$ in $\mathcal{GB}_n^4$ and $f = a_1(\xx)+2a_2(\xx)+2^2a_3(\xx)$ in $\mathcal{GB}_n^8$ were completely
characterized in terms of properties of the Boolean functions $a_i(\xx)$.
In particular, relations between gbentness of $f$ and bentness of associated Boolean functions have been investigated.

In this paper we analyze relations between gbent functions in $\mathcal{GB}_n^{2^k}$ and associated (generalized) Boolean functions.
In Section~\ref{secprel}, some preliminary results are shown, which we will use in the sequel. In particular we show that every gbent function in
$\mathcal{GB}_n^{2^k}$ is regular. In Section~\ref{sec3} we comprehensively characterize gbent functions in $\mathcal{GB}_n^{16}$ in terms of associated Boolean functions,
as well as in terms of associated functions in $\mathcal{GB}_n^4$, which extends results of~\cite{ST09,smgs} on gbent functions in $\mathcal{GB}_n^4$ and $\mathcal{GB}_n^8$.
Furthermore we analyze generalized semibent functions in $\mathcal{GB}_n^{16}$ in terms of associated Boolean functions.
We show in Section~\ref{sec4} that the Gray image of a gbent function in $\mathcal{GB}_n^{8}$, $\mathcal{GB}_n^{16}$ is semibent, respectively, 3-plateaued, which also
extends a result in~\cite{smgs}.
Finally, in Section~\ref{sec3.0} we analyze the relations between gbent functions in $\mathcal{GB}_n^{2^k}$ and their components for general $k>1$.


\section{Preliminaries}
\label{secprel}

We start with a theorem about the regularity of gbent functions, which is also of independent interest. We prove the result by modifying a method of Kumar, Scholtz and Welch~\cite{KSW85}.
\begin{theorem}
\label{propreg}
All gbent functions in $\cGB_n^{2^k}$ are regular.
\end{theorem}
\begin{proof}
If $k=1$, the result is known, as we are dealing with classical bent functions. Let $k\geq 2$.
Let  $\zeta=e^{\frac{2\pi i}{2^k}}$ be a $2^k$-primitive root of unity. It is known that $\Z[\zeta]$ is the ring of algebraic integers in the cyclotomic field~$\Q(\zeta)$. We recall some facts
from~\cite{KSW85} (we change the notations slightly).  The decomposition for the ideal 
generated by $2$ in $\Z[\zeta]$ has the form $\displaystyle \langle 2\rangle= P^{2^{k-1}}$, where $P=\langle 1-\zeta\rangle$ is a prime ideal in $\Z[\zeta]$. The decomposition group
\[
G_2=\{\sigma\ \text{in the Galois group of $\Q(\zeta)/\Q$} \mid \sigma(P)=P\}
 \]
 contains also the conjugation isomorphism $\sigma^*(z)=z^{-1}$ (Proposition 2 in~\cite{KSW85}).
 Observe that $\cH_f^{(2^k)}(\uu )\overline{\cH_f^{(2^k)}(\uu )}=2^k$. Now, as in Property 7 of~\cite{KSW85}, observing that our generalized Walsh transform is simply $S(f,2^{k-1} \uu)$
 (in the notations of Kumar et al.~\cite{KSW85}; $\uu$ is a binary vector in our case), then $\cH_f^{(2^k)}(\uu )$ and $\overline{\cH_f^{(2^k)}(\uu)}$ will generate the same ideal in
 $\Z[\zeta]$ and so, $2^{-k}(\cH_f^{(2^k)}(\uu ))^2$ is a unit, and consequently, $2^{-k/2} \cH_f^{(2^k)}(\uu )$ is an algebraic integer.  Therefore, by Proposition 1 of~\cite{KSW85}
 (which, in fact, it is an old result of Kronecker from 1857), $2^{-k/2} \cH_f^{(2^k)}(\uu )$ must be a root of unity.  That alone would still not be enough to show regularity since
 this root of unity may be in a cyclotomic field outside $\Q(\zeta)$, however, that is not the case here, since the Gauss quadratic sum
 $\displaystyle G(2^k)=\sum_{i=0}^{2^k-1} \zeta^{i^2}= 2^{k/2}(1+i)$ and so, $\sqrt{2}\in\Q(\zeta)$.
\end{proof}

From the definition of a Boolean bent function via the Walsh-Hadamard transform we immediately obtain the following equivalent definition,
where we denote the support of a Boolean function $f$ by ${\rm supp}(f) := \{\xx\in\V_n\;:\;f(\xx) = 1\}$:
A Boolean function $f:\V_n\rightarrow\F_2$ is bent if and only if for every $\uu\in\V_n$ the function $f_{\uu}(\xx) := f(\xx)\+\uu\cdot\xx$ satisfies
$|{\rm supp}(f_{\uu})| = 2^{n-1}\pm 2^{n/2}$.
Our next target is to show an analog description for gbent functions. We use the following lemma.
\begin{lemma}
\label{L4Ny}
Let $q=2^k$, $k>1$, $\zeta=e^{2\pi i/q}$.
If $\rho_l\in\Q$, $0 \le l\le q-1$ and $\sum_{l=0}^{q-1}\rho_l\zeta^l = r$  is rational,
then $\rho_j = \rho_{2^{k-1}+j}$, for $1\le j \le 2^{k-1}-1$.
\end{lemma}
\begin{proof}
Since $\zeta^{2^{k-1}+l} = -\zeta^l$ for $0\le l \le 2^{k-1}-1$, we can write every element $z$
of the cyclotomic field $\Q(\zeta)$ as
\[ z = \sum_{l=0}^{2^{k-1}-1}\lambda_l\zeta^l,\,\lambda_l\in\Q, 0\le l\le 2^{k-1}-1. \]
As $[\Q(\zeta):\Q] = \varphi(q) = 2^{k-1}$ ($\varphi$ is Euler's totient function), the set $\{1,\zeta,\ldots,\zeta^{2^{k-1}-1}\}$ is a basis of
$\Q(\zeta)$. Since
\[ 0 = \sum_{l=0}^{q-1}\rho_l\zeta^l - r = (\rho_0-\rho_{2^{k-1}}-r) + \sum_{l=1}^{{2^{k-1}-1}}(\rho_j - \rho_{2^{k-1}+j})\zeta^l. \]
the assertion of the lemma follows.
\end{proof}
\begin{proposition}
\label{valdis}
Let $n=2m$ be even, and for a function $f:\V_n\rightarrow\Z_{2^k}$ and $\uu\in\V_n$, let $f_{\uu}(\xx) = f(\xx)+2^{k-1}(\uu\cdot\xx)$, and let
$b_j^{(\uu)} = |\xx\in\V_n\;:\;f_{\uu}(\xx) = j\}|$, $0\le j\le 2^k-1$. Then $f$ is gbent if and only if for all $\uu\in\V_n$
there exists an integer $\rho_{\uu}$, $0\le \rho_{\uu} \le 2^{k-1}-1$ such that
\[ b^{(\uu)}_{2^{k-1}+\rho_{\uu}} = b^{(\uu)}_{\rho_{\uu}}\pm 2^m\; \mbox{and}\; b^{(\uu)}_{2^{k-1}+j} = b^{(\uu)}_j,\,\mbox{for}\; 0\le j \le 2^{k-1}-1, j\ne \rho_{\uu}. \]
\end{proposition}
\begin{proof}
First suppose that $f$ is gbent. Then by Theorem~\ref{propreg}, $f$ is a regular gbent function. Hence
\[ \mathcal{H}_f(\uu) = \sum_{\xx\in\V_n}\zeta^{f(\xx)}(-1)^{\uu\cdot\xx} = \sum_{\xx\in\V_n}\zeta^{f(\xx)+2^{k-1}(\uu\cdot\xx)}
= \mathcal{H}_{f_{\uu}}(0) = \sum_{j=0}^{2^k-1}b^{(\uu)}_j\zeta^j = 2^m\zeta^r \]
for some $0\le r\le 2^k-1$. With $\rho_{\uu} = r$ if $0\le r\le 2^{k-1}-1$, and $\rho_{\uu} = r-2^{k-1}$ otherwise, the claim follows from Lemma~\ref{L4Ny}.

The converse statement is verified in a straightforward manner.
\end{proof}

We will frequently use the following easily verified identity.
\begin{lemma}
\label{lem:zs}
Let $z$ be a complex number.
If $s\in\{0,1\}$,  then
\begin{equation*}
z^s=\frac{1+(-1)^s}{2}+\frac{1-(-1)^s}{2}z.
\end{equation*}
\end{lemma}

Let $f\in\mathcal{GB}_n^{2^k}$ be given as $f(\xx) = a_1(\xx)+2a_2(\xx)+\cdots+2^{k-1}a_k(\xx)$, $a_i\in\mathcal{B}_n$, $1\le i\le k$.
With the very general Theorem~2 of \cite{smgs}, one can express the generalized Walsh-Hadamard transform $\mathcal{H}_f(\uu)$ in terms
of the Walsh-Hadamard transforms of Boolean functions obtained as sums of $a_i(\xx)$, $i\in\{1,\ldots,k\}$. As one may expect, this representation
in its generality is not quite explicit. As special cases we represent $\mathcal{H}_f^{(2^k)}(\uu)$ in terms of the Walsh-Hadamard transforms of
the Boolean functions $c_1a_1(\xx) \+ \cdots \+ c_{k-1}a_{k-1}(\xx) \+ a_k(\xx)$, $c_i\in\F_2$, for $k = 2,3,4$ explicitly in the following lemma.
For $k=2$ and $k=3$ see also Lemma~3.1 in~\cite{ST09} and Lemma~17 in~\cite{smgs}.
\begin{lemma}
\label{lem-Hf}
The following statements are true:
\begin{itemize}
\item[$(i)$]
Let $f(\xx)=a_1(\xx)+2a_2(\xx)\in\cGB_n^4$ with $a_1,a_2\in\cB_n$. The generalized Walsh-Hadamard transform of $f$ is given by
\[
2\cH_f^{(4)}(\uu )=\left(\cW_{a_2}(\uu )+\cW_{a_1\+a_2}(\uu ) \right)+i \left(\cW_{a_2}(\uu )-\cW_{a_1\+a_2}(\uu ) \right).
\]
\item[$(ii)$] Let $f(\xx)=a_1(\xx)+2a_2(\xx)=2^2a_3(\xx)\in\cGB_n^8$ with $a_1,a_2,a_3\in\cB_n$. The generalized Walsh-Hadamard
transform of $f$ is given by
\[ 4\cH_f^{(8)}(\uu )= \alpha_0\mathcal{W}_{a_3}(\uu) + \alpha_{1}\mathcal{W}_{a_1\+a_3}(\uu) + \alpha_{2}\mathcal{W}_{a_2\+a_3}(\uu) + \alpha_{12}\mathcal{W}_{a_1\+a_2\+a_3}(\uu), \]
where $\alpha_0 = 1+(1+\sqrt{2})i$, $\alpha_1 = 1+(1-\sqrt{2})i$, $\alpha_2 = 1+\sqrt{2}-i$, $\alpha_{12} = 1-\sqrt{2}-i$.
\item[$(iii)$]
The generalized Walsh-Hadamard transform of $f\in\mathcal{GB}_n^{16}$ with
$f(\xx)=a_1(\xx)+2a_2(\xx)+2^2a_3(\xx)+2^3 a_4(\xx)$, $a_i\in\cB_n$, $1\leq i\leq 4$, is given by
\begin{equation*}
\begin{split}
8\cH_f^{(16)}(\uu )&=\alpha_0 \cW_{a_4}(\uu )+\alpha_1  \cW_{a_1\+a_4}(\uu )+\alpha_2  \cW_{a_2\+a_4}(\uu )\\
&+\alpha_3  \cW_{a_3\+a_4}(\uu ) +\alpha_{12}  \cW_{a_1\+a_2\+a_4}(\uu )+
\alpha_{13}  \cW_{a_1\+a_3\+a_4}(\uu )\\
& +\alpha_{23}  \cW_{a_2\+a_3\+a_4}(\uu )+\alpha_{123}  \cW_{a_1\+a_2\+a_3\+a_4}(\uu ),
\end{split}
\end{equation*}
where
\begin{align*}
\alpha_0 &=(1+i)(1 + \zeta + \zeta^2 + \zeta^3)=1+i \left(1+\sqrt{2}+\sqrt{2 \left(2+\sqrt{2}\right)}\right), \\
\alpha_1 &=(1+i)(1 - \zeta + \zeta^2 - \zeta^3)=1+i \left(1+\sqrt{2}-\sqrt{2 \left(2+\sqrt{2}\right)}\right), \\
\alpha_2 &=(1+i)(1 + \zeta - \zeta^2 - \zeta^3)=1+\sqrt{2 \left(2-\sqrt{2}\right)}+i \left(1-\sqrt{2}\right), \\
\alpha_3 &=(1-i)(1 + \zeta + \zeta^2 + \zeta^3)=1+\sqrt{2}+\sqrt{2 \left(2+\sqrt{2}\right)}-i, \\
\alpha_{12} &=(1+i)(1 - \zeta - \zeta^2 + \zeta^3)=1-\sqrt{2 \left(2-\sqrt{2}\right)}+i \left(1-\sqrt{2}\right),\\
\alpha_{13} &=(1-i)(1 - \zeta + \zeta^2 - \zeta^3)=1+\sqrt{2}-\sqrt{2 \left(2+\sqrt{2}\right)}-i,\\
\alpha_{23} &=(1-i)(1 + \zeta - \zeta^2 - \zeta^3)=1-\sqrt{2}-i \left(1+\sqrt{2 \left(2-\sqrt{2}\right)}\right),\\
\alpha_{123} &=(1-i)(1 - \zeta - \zeta^2 + \zeta^3)=1-\sqrt{2}-i \left(1-\sqrt{2 \left(2-\sqrt{2}\right)}\right).
\end{align*}
%
\end{itemize}
\end{lemma}
\begin{proof}
One can show all cases by straightforward direct calculations (using Lemma~\ref{lem:zs}) or by applying ~\cite[Theorem 2]{smgs}.
For $k=4$ both approaches are quite cumbersome. We will only perform the calculations for $(i)$, where $k=2$.

By Lemma~\ref{lem:zs}, we write
\begin{align*}
2\cH_f^{(4)}(\uu )&=2\sum_{\xx\in\BBF_2^n} i^{a_1(\xx)+2a_2(\xx)} (-1)^{\uu\cdot\xx}\\
&=\sum_{\xx\in\BBF_2^n}\left((1+(-1)^{a_1(\xx)})+i(1-(-1)^{a_1(\xx)})\right) (-1)^{a_2(\xx)} (-1)^{\uu\cdot\xx}\\
&=(1+i)\sum_{\xx\in\BBF_2^n}  (-1)^{a_2(\xx)} (-1)^{\uu\cdot\xx}+(1-i)\sum_{\xx\in\BBF_2^n}  (-1)^{a_1(\xx)\+a_2(\xx)} (-1)^{\uu\cdot\xx}\\
&=(1+i) \cW_{a_2}(\uu )+(1-i) \cW_{a_1\+a_2}(\uu )\\
&=\left(\cW_{a_2}(\uu )+\cW_{a_1\+a_2}(\uu ) \right)+i \left(\cW_{a_2}(\uu )-\cW_{a_1\+a_2}(\uu ) \right).
\end{align*}
\end{proof}

\begin{lemma}
\label{lem-ind}
The set $\left\{ 1,\sqrt{2},\sqrt{2-\sqrt{2}},\sqrt{2+\sqrt{2}} \right\}$ is linear independent over $\mathbb{Z}$ (certainly, over $\mathbb{Q}$, as well).
\end{lemma}
\begin{proof}
%
Assume that there is a nontrivial linear relation of the form
\begin{equation*}
\label{eq-lin}
a+b\sqrt{2}+c \sqrt{2-\sqrt{2}}+d\sqrt{2+\sqrt{2}}=0, \ a,b,c,d\in\mathbb{Z}.
\end{equation*}
Without loss of generality, we assume that $\gcd(a,b,c,d)=1$.
By moving the last two terms to the right side and squaring, we obtain
\[
a^2+2b^2+2ab\sqrt{2}=c^2(2-\sqrt{2})+d^2(2+\sqrt{2})+2cd\sqrt{2},
\]
that is,
$
\sqrt{2} (c^2-d^2-2cd+2ab)=2c^2+2d^2-a^2-2b^2,
$
and so,
\begin{align*}
&c^2-d^2-2cd+2ab=0,\\
&2c^2+2d^2-a^2-2b^2=0.
\end{align*}
The first equation shows that $c,d$ must have the same parity, and the second shows that $a\equiv 0\pmod 2$, say $a=2^r a_1$, $r\geq 1,a_1\in\Z$.
Consequently, $c^2-d^2+2ab\equiv 0\pmod 4$, which implies that $2cd \equiv 0\bmod 4$. Therefore $c\equiv d \equiv 0 \bmod 2$,
say $c=2^t c_1,d=2^u d_1$, $t\geq 1,u\geq 1$, $c_1,d_1\in\Z$, and by the second equation we have $2^{2t}c_1^2+2^{2u}d_1^2-2^{2r-1}a_1^2-b^2=0$.
From $\gcd(a,b,c,d)=1$, which implies that $b$ is odd, we see that $2b^2 \equiv 2\pmod 4$. This yields a contradiction since
$2^{2t}c_1^2 \equiv 2^{2u}d_1^2 \equiv 2^{2r}a_1^2 \equiv 0 \pmod 4$.
\end{proof}

\begin{remark}
The set $\{1,\sqrt{2},\alpha_1=\sqrt{2-\sqrt{2}},\alpha_2=\sqrt{2+\sqrt{2}}\}$ is actually a basis of $K=\Q(\sqrt{2},\alpha_1)$ over $\Q$,
the splitting field of $(x^2-2)(x^4-4x^2+2)$. Note that $\alpha_1\alpha_2 = \sqrt{2}$.
%
\end{remark}

The following lemma is used several times in the next few sections, and we find it is worth to be noted on its own.
\begin{lemma}
\label{HHH}
Let $f\in\mathcal{GB}_n^{2^k}$ with $f(\xx)=g(\xx)+2h(\xx), g\in\cB_n,h\in\cGB_n^{2^{k-1}}$. Then
\begin{equation}
\label{eq:gb2k}
2\cH_f^{(2^k)}(\uu )=(1+\zeta_{2^k}) \cH_h^{(2^{k-1})}(\uu )+(1-\zeta_{2^k}) \cH_{h+2^{k-2}g}^{(2^{k-1})}(\uu ).
\end{equation}
\end{lemma}
\begin{proof}
Using Lemma~\ref{lem:zs}, we write
\begin{align*}
2\cH_f^{(2^k)}(\uu )&=2\sum_{\xx\in\BBF_2^n} \zeta_{2^k}^{g(\xx)}\zeta_{2^{k-1}}^{h(\xx)} (-1)^{\uu\cdot\xx}\\
&=\sum_{\xx\in\BBF_2^n} \left(1+(-1)^{g(\xx)}+(1-(-1)^{g(\xx)})\zeta_{2^k}\right)\zeta_{2^{k-1}}^{h(\xx)} (-1)^{\uu\cdot\xx}\\
&=(1+\zeta_{2^k}) \cH_h^{(2^{k-1})}(\uu )+(1-\zeta_{2^k}) \cH_{h+2^{k-2}g}^{(2^{k-1})}(\uu ).
\end{align*}
\end{proof}

\section{Complete characterization of generalized bent and semibent functions in $\cGB_n^{16}$}
\label{sec3}

In this section we characterize gbent functions in $\cGB_n^{16}$ in several ways.
We write $f \in \cGB_n^{16}$ as
\begin{align*}
f(\xx) &= a_1(\xx) + 2a_2(\xx) + 2^2a_3(\xx) + 2^3a_4(\xx) \\
& = b_1(\xx) + 2^2b_2(\xx) = a_1(\xx) + 2d(\xx),
\end{align*}
where $a_i(\xx)\in\mathcal{B}_n$, $i=1,2,3,4$, $b_1(\xx)=a_1(\xx)+2a_2(\xx)$, $b_2(\xx)=a_3(\xx)+2a_4(\xx)$ are in $\mathcal{GB}_n^{4}$,
and $d(\xx)=a_2(\xx)+2a_3(\xx)+2^2a_4(\xx)\in\mathcal{GB}_n^{8}$.

Our objective is to show necessary {\it and} sufficient conditions on the components $a_1,a_2,a_3,a_4,b_1,b_2,d$ for the
gbentness of $f$. For the conditions on $a_1$ and $d$ for the gbentness of $a_1(\xx)+2d(\xx)$ when $n$ is even, we can
refer to our general result in Theorem \ref{k,k-1Thm} in Section \ref{sec3.0}. For the special case of gbent functions
in $\mathcal{GB}_n^{16}$, $n$ even, it states that $f(\xx) = a_1(\xx)+2d(\xx)$ is gbent if and only if $d$ and $d+4a_1$ are
gbent in $\cGB_n^8$ and $\Im\left(\overline{\cH_d^{(8)}(\uu )}\cH_{d+4a_1}^{(8)}(\uu )\right)=0$ for all $\uu\in\V_n$.

The first target in this section is to show necessary and sufficient conditions on the Boolean functions $a_1,a_2,a_3,a_4$ for $f$ to be
gbent $\mathcal{GB}_n^{16}$ for even as well as for odd $n$. Secondly, necessary and sufficient conditions on the functions $b_1,b_2$ for the gbentness
of $f$ are given. This complete characterization of gbent functions in $\mathcal{GB}_n^{16}$ extends results in \cite{ST09,smgs} on gbent functions in
$\mathcal{GB}_n^{4}$ and $\mathcal{GB}_n^{8}$. Finally we also characterize gsemibent functions $f\in\mathcal{GB}_n^{16}$ in terms of $a_1,a_2,a_3,a_4$.

\begin{theorem}
\label{thm-gengb16}
Suppose that $f(\xx)=a_1(\xx)+2a_2(\xx)+2^2 a_3(\xx)+2^3a_4(\xx)$,
$a_i\in\cB_n$, $1\leq i\leq 4$.
Then $f$ is gbent in $\cGB_n^{16}$ if and only if the conditions $(i)$ (if $n$ is even), or $(ii)$ (if $n$ is odd) hold:
\begin{sloppypar}
  \begin{enumerate}
\item[$(i)$] For all $c_i \in \F_2$, $i=1,2,3$, the Boolean function $c_1a_1\+ c_2a_2\+ c_3a_3\+ a_4$ is bent, and for all $\uu \in \V_n$ we have
\begin{align*}
& \mathcal{W}_{a_4}(\uu)\mathcal{W}_{a_2\+a_4}(\uu) =
\mathcal{W}_{a_3\+a_4}(\uu)\mathcal{W}_{a_2\+a_3\+a_4}(\uu) \\
& = \mathcal{W}_{a_1\+a_4}(\uu)\mathcal{W}_{a_1\+a_2\+a_4}(\uu)
 =\mathcal{W}_{a_1\+a_3\+a_4}(\uu)\mathcal{W}_{a_1\+a_2\+a_3\+a_4}(\uu),\;\mbox{and}\\
& \mathcal{W}_{a_4}(\uu)\mathcal{W}_{a_3\+a_4}(\uu) = \mathcal{W}_{a_1\+a_4}(\uu)\mathcal{W}_{a_1\+a_3\+a_4}(\uu).
\end{align*}
\item[$(ii)$]
For all $c_i \in \F_2$, $i=1,2,3$, the Boolean function $c_1a_1\+ c_2a_2\+ c_3a_3\+ a_4$ is semibent, and for all $\uu \in \V_n$ we either have
\begin{align*}
 & \cW_{a_4}(\uu )\cW_{a_2\+ a_4}(\uu )=\cW_{a_1\+ a_4}(\uu ) \cW_{a_1\+a_2\+a_4}(\uu ) =\pm 2^{n+1}\;\mbox{and} \\
 & \cW_{a_3\+ a_4}(\uu )=\cW_{a_2\+ a_3\+ a_4}(\uu )=\cW_{a_1\+a_3\+ a_4}(\uu )=\cW_{a_1\+a_2\+ a_3\+ a_4}(\uu )=0,
\end{align*}
or
\begin{align*}
& \cW_{a_2\+ a_4}(\uu ) = \cW_{a_4}(\uu )=\cW_{a_1\+ a_4}(\uu ) = \cW_{a_1\+a_2\+a_4}(\uu )=0\;\mbox{and} \\
& \cW_{a_3\+a_4}(\uu) \cW_{a_2\+a_3\+a_4}(\uu) = \cW_{a_1\+a_3\+a_4}(\uu)\cW_{a_1\+a_2\+a_3\+a_4}(\uu) = \pm 2^{n+1}.
\end{align*}
  \end{enumerate}
  \end{sloppypar}
\end{theorem}
\begin{proof}
Let $\uu\in \V_n$. For $k=4$,
Equation $(\ref{eq:gb2k})$ equals
\[ 2\cH_f^{(16)}(\uu)=(1+\zeta_{16}) \cH_d^{(8)}(\uu) +(1-\zeta_{16})\cH_{d+4a_1}^{(8)}(\uu). \]
Taking norms
and squaring, in this case we get
{\small
\begin{equation*}
\begin{split}
4|\cH_f^{(16)}(\uu )|^2&=\left((1+\zeta_{16}) \cH_d^{(8)}(\uu ) +(1-\zeta_{16})\cH_{d+4a_1}^{(8)}(\uu )\right)\\
&\qquad \times \left((1+\overline{\zeta_{16}}) \overline{\cH_d^{(8)}(\uu)} +(1-\overline{\zeta_{16}})\overline{\cH_{d+4a_1}^{(8)}(\uu)}\right)\\
&=(1+\zeta_{16})(1+\overline{\zeta_{16}})|\cH_d^{(8)}(\uu )|^2+(1-\zeta_{16})(1-\overline{\zeta_{16}})||\cH_{d+4a_1}^{(8)}(\uu )|^2\\
&\qquad+(1+\zeta_{16})(1-\overline{\zeta_{16}}) \cH_d^{(8)}(\uu )\overline{\cH_{d+4a_1}^{(8)}(\uu)}\\
&\qquad +(1-\zeta_{16})(1+\overline{\zeta_{16}})\bar\cH_d^{(8)}(\uu )\cH_{d+4a_1}^{(8)}(\uu )\\
&=(2+\sqrt{2+\sqrt{2}}) |\cH_d^{(8)}(\uu )|^2+(2-\sqrt{2+\sqrt{2}}) |\cH_{d+4a_1}^{(8)}(\uu )|^2\\
&\qquad +2\sqrt{2-\sqrt{2}}\ \Im\left(\overline{\cH_d^{(8)}(\uu )}\cH_{d+4a_1}^{(8)}(\uu )\right).
\end{split}
\end{equation*}
Equivalently,
\begin{align}
\label{eq-hf-gb8} \nonumber
16\sqrt{2}|\cH_f^{(16)}(\uu )|^2 &= (2+\sqrt{2+\sqrt{2}}) 4\sqrt{2}|\cH_d^{(8)}(\uu )|^2+(2-\sqrt{2+\sqrt{2}}) 4\sqrt{2}|\cH_{d+4a_1}^{(8)}(\uu )|^2\\
&+8\sqrt{4-2\sqrt{2}}\ \Im\left(\overline{\cH_d^{(8)}(\uu )}\cH_{d+4a_1}^{(8)}(\uu )\right).
\end{align}
We denote by $A,C,D,W$ the Walsh-Hadamard transforms $\cW_{a_4}(\uu )$, $\cW_{a_2\+a_4}(\uu )$, $\cW_{a_3\+a_4}(\uu )$, $\cW_{a_2\+a_3\+a_4}(\uu )$ (in that order).
We denote by $B,X,Y,Z$ the Walsh-Hadamard transforms $\cW_{a_1\+a_4}(\uu )$, $\cW_{a_1\+a_2\+a_4}(\uu )$, $\cW_{a_1\+a_3\+a_4}(\uu )$, $\cW_{a_1\+a_2\+a_3\+a_4}(\uu )$ (in that order).
By Lemma \ref{lem-Hf}(ii),
we know that the generalized Walsh-Hadamard transform of any function in $\cGB_n^8$, say $d$ and $d+4a_1$ with $d = a_2+2a_3+2^2a_4$, is of the form
\begin{align*}
4\cH_d^{(8)}(\uu )&= \alpha_0 A+\alpha_1 C+\alpha_2 D+\alpha_3 W,\\
4\cH_{d+4a_1}^{(8)}(\uu )&= \alpha_0 B+\alpha_1 X+\alpha_2 Y+\alpha_3 Z,
\end{align*}
where $\alpha_0= 1+(1+\sqrt{2})i$, $\alpha_1=1+(1-\sqrt{2})i$, $\alpha_2=1+\sqrt{2}-i$, $\alpha_3=1-\sqrt{2}-i$, and
moreover that (see also ~\cite[Corollary 18]{smgs}),
{\small
\begin{align}
\label{E1}
4\sqrt{2}|\cH_d^{(8)}(\uu )|^2&=A^2-C^2+2 C D+D^2-2 A W-W^2+\sqrt{2} (A^2+C^2+D^2+W^2)\\ \nonumber
4\sqrt{2}|\cH_{d+4a_1}^{(8)}(\uu )|^2&=B^2-X^2+2 X Y+Y^2-2 B Z -Z^2+\sqrt{2} (B^2+X^2+Y^2+Z^2).
\end{align}
}
Furthermore, with straightforward computations we get
{\small
\allowdisplaybreaks[3]
\begin{equation}
\label{E2}
\begin{split}
 &8\sqrt{4-2\sqrt{2}}\ \Im\left(\overline{\cH_b^{(8)}(\uu )}\cH_{b+4a_1}^{(8)}(\uu )\right)\\
 &=2\sqrt{2-\sqrt{2}}\left(
 \sqrt{2}\, (B D + W X - A Y - C Z)\right.\\
&\left.\quad +B C + B D - A X - W X - A Y + W Y + C Z - D Z\right)\\
&=2\sqrt{2-\sqrt{2}}\, (B C + B D - A X - W X - A Y + W Y + C Z - D Z)\\
&\quad+2\sqrt{4-2\sqrt{2}}\, (B D + W X - A Y - C Z)
\end{split}
\end{equation}
}
With \eqref{E1} and \eqref{E2} for Equation \eqref{eq-hf-gb8} we obtain
{\small
\allowdisplaybreaks[3]
\begin{align}
\label{forsystem}
\nonumber
&16\sqrt{2}|\cH_f^{(16)}(\uu )|^2\\ \nonumber
&= 2   (A^2 + B^2 - C^2 + 2 C D + D^2 - 2 A W - W^2 - X^2 + 2 X Y + Y^2 - 2 B Z - Z^2)\\ \nonumber
&\quad + 2   \sqrt{2} (A^2+B^2+C^2+D^2+W^2+X^2+Y^2+Z^2)\\
  &\quad +\sqrt{2-\sqrt{2}}\, (A^2 - B^2 + 2 B C + C^2 + D^2 + W^2 - 2 A X - 4 W X - X^2 \\ \nonumber
   &\qquad  + 2 W Y - Y^2 + 4 C Z - 2 D Z - Z^2)\\ \nonumber
    &\quad +2\sqrt{2+\sqrt{2}}\, (A^2 - B^2 + B D + C D + D^2 - A W + W X - A Y - X Y\\ \nonumber
     &\qquad  - Y^2 + B Z - C Z).
\end{align}
}
Now suppose that $f$ is gbent in $\cGB_n^{16}$, i.e., $|\cH_f^{(16)}(\uu )|^2 = 2^n$. By Lemma~\ref{lem-ind}, $1,\sqrt{2},\sqrt{2-\sqrt{2}},\sqrt{2+\sqrt{2}}$
are $\mathbb{Z}$-linearly independent, and hence we arrive at the following a system of equations (in the variables $A,B,C,D,X,Y,Z,W$) with solutions in $\mathbb{Z}$:
{\small
\allowdisplaybreaks[3]
\begin{equation}
\label{eq:2ndsystem}
\begin{split}
&A^2 + B^2 + C^2 + D^2 + W^2 + X^2 + Y^2 + Z^2=2^{n+3}\\
&A^2 + B^2 - C^2 + 2 C D + D^2 - 2 A W - W^2 - X^2 + 2 X Y\\
   &\qquad\qquad\qquad\qquad  + Y^2 - 2 B Z - Z^2=0\\
&A^2 - B^2 + 2 B C + C^2 + D^2 + W^2 - 2 A X - 4 W X - X^2 \\
   &\qquad\qquad\qquad\qquad\qquad + 2 W Y - Y^2 + 4 C Z - 2 D Z - Z^2=0\\
 & A^2 - B^2 + B D + C D + D^2 - A W + W X - A Y - X Y\\
   &\qquad\qquad\qquad\qquad\qquad  - Y^2 + B Z - C Z=0.
\end{split}
\end{equation}
}}

Let $2^t$ be the largest power of $2$ which divides all, $A,B,C,D,X,Y,Z$ and $W$, i.e., $A=2^t A_1$, etc., with at least one of the $A_1,B_1,\ldots$ being odd.
First, if $n$ is even and $t>\frac{n}{2}$, then $t=\frac{n}{2}+1$ only.  Dividing by $2^{2t}$, the first equation of~\eqref{eq:2ndsystem} becomes
$A_1^2 + B_1^2 + C_1^2 + D_1^2 + W_1^2 + X_1^2 + Y_1^2 + Z_1^2=2$, which gives the solution $(\pm 1,\pm 1,0,0,0,0,0,0)$ (and permutations of these values).
However, a simple computation reveals that none of these possibilities also satisfies the last three equations of~\eqref{eq:2ndsystem}.  If $n$ is odd and
$t>\frac{n+1}{2}$, then $t$ must be $t=\frac{n+3}{2}$, but this implies that only one value out of $A,B,\ldots$ is nonzero and again, that is impossible to
satisfy the last three equations of~\eqref{eq:2ndsystem}.
Assume now that $t<\frac{n}{2}$. The first equation of~\eqref{eq:2ndsystem} becomes
$A_1^2 + B_1^2 + C_1^2 + D_1^2 + W_1^2 + X_1^2 + Y_1^2 + Z_1^2=2^{n+3-2t}$, which is divisible by $2^5$ (when $n$ is even, since $t\leq \frac{n-2}{2}$),
respectively $2^4$ (when $n$ is odd, since $t\leq \frac{n-1}{2}$). If $n$ is even, this can only happen if  $A_1,B_1,\ldots,$ are all even, that is, $\equiv 0,2,4,6\pmod 8$,
but that contradicts our assumption that $t$ is the largest power of 2 dividing $A,B,\ldots$. If $n$ is odd and $t\leq \frac{n-3}{2}$, the previous argument would work, and
if $t=\frac{n-1}{2}$, then $A_1^2 + B_1^2 + C_1^2 + D_1^2 + W_1^2 + X_1^2 + Y_1^2 + Z_1^2=16$.
 One can certainly argue exhaustively, or by considering every residues for $A_1,B_1,\ldots,$ modulo 4 and imposing the condition that the 2nd, 3rd, 4th equations of our system also must be 0 modulo 16, we only get possibilities
 $(0, 0, 2, 2, 0, 0, 2, 2)$, $(0, 2, 0, 2, 0, 2, 0, 2)$, $(0, 2, 2, 0, 0, 2, 2, 0)$, $(2, 0, 0, 2, 2, 0, 0, 2)$,
 $(2, 0, 2, 0, 2, 0, 2, 0)$, $(2, 2, 0, 0, 2, 2, 0, 0)$ for $(A_1,B_1,\ldots)$ modulo 4, but that implies that all $A_1,B_1,\ldots$ are even, but that contradicts our assumption that $t$ is the largest power of 2 dividing $\gcd(A,B,\ldots,)$.
This shows that the only possibility is $2^t = 2^{n/2}$ if $n$ is even, and $2^t = 2^{(n+1)/2}$ if $n$ is odd.

Thus, one needs to find integer solutions for the equation $A_1^2 + B_1^2 + C_1^2 + D_1^2 + \cW_1^2 + X_1^2 + Y_1^2 + Z_1^2=8$, for $n$ even, or
$A_1^2 + B_1^2 + C_1^2 + D_1^2 + \cW_1^2 + X_1^2 + Y_1^2 + Z_1^2=4$ for $n$ odd, which also satisfy the last three equations in~\eqref{eq:2ndsystem}.
Again, Mathematica renders the following:  if $n$ is even,  then $2^{-\frac{n}{2}}(A, C, D, W, B, X, Y, Z)$ (note the order) is one of
\allowdisplaybreaks[3]
\begin{align}
\label{eq-valuesH}
&(-1, -1, -1, -1, -1, -1, -1, -1), &(-1, -1, 1, 1, -1, -1, 1, 1),\nonumber\\
&(-1,  1, -1, 1, -1, 1, -1, 1), &(-1, 1, 1, -1, -1, 1,  1, -1),\nonumber\\
&(-1, -1, -1, -1, 1, 1, 1, 1), &(-1, -1, 1, 1, 1,  1, -1, -1), \nonumber\\
&(-1, 1, -1, 1, 1, -1, 1, -1), &(-1, 1, 1, -1, 1, -1, -1,  1),\\
&(1, -1, -1, 1, -1, 1, 1, -1), &(1, -1, 1, -1, -1, 1, -1, 1),\nonumber\\
&(1, 1, -1, -1, -1, -1, 1, 1), &(1, 1, 1, 1, -1, -1, -1, -1), \nonumber\\
&(1, -1, -1,  1, 1, -1, -1, 1), &(1, -1, 1, -1, 1, -1, 1, -1),\nonumber\\
&(1, 1, -1, -1, 1,  1, -1, -1), &(1, 1, 1, 1, 1, 1, 1, 1)\nonumber
\end{align}
and, if $n$ is odd, then $2^{-\frac{n+1}{2}}(A, C, D, W, B, X, Y, Z)$ is one of
\allowdisplaybreaks[3]
\begin{align*}
& (-1, -1, 0, 0, -1, -1, 0, 0), &(-1, 1, 0, 0, -1, 1, 0, 0), \\
& (-1, -1, 0,  0, 1, 1, 0, 0), &(-1, 1, 0, 0, 1, -1, 0, 0),\\
& (0, 0, -1, -1, 0,  0, -1, -1), &(0, 0, -1, 1, 0, 0, -1, 1),\\
& (0, 0, -1, -1, 0, 0, 1, 1), &(0, 0, -1, 1, 0, 0, 1, -1), \\
& (0, 0, 1, -1, 0, 0, -1, 1), &(0, 0, 1, 1, 0, 0, -1, -1),\\
&(0, 0, 1, -1, 0, 0, 1, -1), &(0, 0, 1, 1, 0, 0,  1, 1), \\
& (1, -1, 0, 0, -1, 1, 0, 0), &(1, 1, 0, 0, -1, -1, 0,  0),\\
& (1, -1, 0, 0, 1, -1, 0, 0), &(1, 1, 0, 0, 1, 1, 0, 0).
 \end{align*}
Therefore, in the case of even $n$ we see that if $f$ is gbent, then  $a_4, a_2\+a_4, a_3\+a_4, a_2\+a_3\+a_4, a_1\+a_4, a_1\+a_2\+a_4, a_1\+a_3\+a_4,a_1\+a_2\+a_3\+a_4$ are all bent in $\cB_n$,
such that $\cW_{a_4}(\uu )\cW_{a_2\+a_4}(\uu )=\cW_{a_3\+a_4}(\uu )\cW_{w_2\+a_3\+a_4}(\uu )=\cW_{a_1\+a_4}(\uu )\cW_{a_1\+a_2\+a_4}(\uu )=\cW_{a_1\+a_3\+a_4}(\uu )\cW_{a_1\+a_2\+a_3\+a_4}(\uu )$
and $\mathcal{W}_{a_4}(\uu)\mathcal{W}_{a_3+a_4}(\uu) = \mathcal{W}_{a_1+a_4}(\uu)\mathcal{W}_{a_1+a_3+a_4}(\uu)$
for every $\uu \in \V_n$.
It is a straightforward computation to see that under these conditions, $f$ is also gbent in $\cGB_n^{16}$.

In the case of odd $n$ we see that if $f$ is gbent, then $a_4, a_2\+a_4, a_3\+a_4, a_2\+a_3\+a_4, a_1\+a_4, a_1\+a_2\+a_4, a_1\+a_3\+a_4,a_1\+a_2\+a_3\+a_4$
are semibent with the extra conditions that $\cW_{a_4}(\uu )\cW_{a_2\+ a_4}(\uu )=\cW_{a_1\+ a_4}(\uu ) \cW_{a_1\+a_2\+a_4}(\uu ) =\pm 2^{n+1}$,
and $\cW_{a_3\+ a_4}(\uu )=\cW_{a_2\+ a_3\+ a_4}(\uu )=\cW_{a_1\+a_3\+ a_4}(\uu )=\cW_{a_1\+a_2\+ a_3\+ a_4}(\uu )=0$, or
$\cW_{a_2\+ a_4}(\uu ) = \cW_{a_4}(\uu )=\cW_{a_1\+ a_4}(\uu ) = \cW_{a_1\+a_2\+a_4}(\uu )=0$, and
$\cW_{a_3\+a_4}(\uu) \\
\cW_{a_2\+a_3\+a_4}(\uu) = \cW_{a_1\+a_3\+a_4}(\uu) \cW_{a_1\+a_2\+a_3\+a_4}(\uu) = \pm 2^{n+1}$, for all $\uu\in\BBF_2^n$.
Reciprocally,  it is straightforward to check that under these conditions, $f$ is also gbent in $\cGB_n^{16}$.
\end{proof}

\begin{remark}
It is not sufficient to only use what is known (see~\textup{\cite{rouse}} for a proof) about the number of solutions for the first equation of~\eqref{eq:2ndsystem}, that is,
$r_8(2^{n+3})=\frac{16(8^{n+4}-15)}{7}$ (of course, counting signs and permutations), since the form of these solutions is not known in general.
\end{remark}

\begin{remark}
If in $f(\xx) = a_1(\xx)+2a_2(\xx)+2^2a_3(\xx)+2^3a_4(\xx)$ some of the Boolean functions $a_i$, $i=1,2,3,4$, are the zero function, and hence the value set of $f$ is restricted
accordingly, then the conditions in Theorem~\textup{\ref{thm-gengb16}} of course will simplify.
For instance we can immediately infer from Theorem~\textup{\ref{thm-gengb16}} that $f(\xx)=a_1(\xx)+2^3a_4(\xx) \in\mathcal{GB}_n^{16}$ is gbent for even $n$ if and only if $a_4,a_1\+a_4$
are both bent, and never gbent for odd $n$. If $a_3 = 0$, i.e., $f$ takes on only values in $\{0,1,\ldots,7\}$, then $f$ is not gbent (which is also quite apparent with a direct argument
via the generalized Walsh-Hadamard transform - and similarly holds for functions in $\mathcal{GB}_n^{2^k}$).
\end{remark}

We next present the connection between gbentness in $\cGB_n^{4}$ and in $\cGB_n^{16}$.
\begin{theorem}
Let $f\in\mathcal{GB}_n^{16}$ with
\[ f(\xx)=a_1(\xx)+2a_2(\xx)+2^2 a_3(\xx)+2^3a_4(\xx)=b_1(\xx)+2^2 b_2(\xx), \]
where $b_1=a_1+2a_2,b_2=a_3+2a_4\in\cGB_n^4$. The function $f$ is gbent in $\cGB_n^{16}$ if and only if $b_2,b_1+b_2,2b_1+b_2,3b_1+b_2$ are gbent in $\cGB_n^4$ with their generalized
Walsh-Hadamard transforms satisfying the following conditions, $(i)$ for $n$ even, respectively, $(ii)$ for $n$ odd, for all $\uu\in\V_n$:
  \begin{enumerate}
\item[$(i)$]
$2^{-n/2}(\cH_{3b_1+b_2}(\uu), \cH_{b_1+b_2}(\uu), \cH_{2b_1+b_2}(\uu), \cH_{b_2}(\uu))$ belongs to one of $(\epsilon,\epsilon,\epsilon,\epsilon)$, $(\epsilon,\epsilon,-\epsilon,-\epsilon)$,
$(\epsilon,-\epsilon,\epsilon i,-\epsilon i)$, $(\epsilon-\epsilon,-\epsilon i,\epsilon i)$, $(\epsilon i, \epsilon i, \epsilon i, \epsilon i)$, $(\epsilon i, \epsilon i,-\epsilon i, -\epsilon i)$,
$(\epsilon i,-\epsilon i,\epsilon,-\epsilon)$, $(-\epsilon i, \epsilon i,-\epsilon, \epsilon)$, where $\epsilon\in\{\pm 1\}$.
\item[$(ii)$] $2^{-(n-1)/2}(\cH_{3b_1+b_2}(\uu), \cH_{b_1+b_2}(\uu), \cH_{2b_1+b_2}(\uu), \cH_{b_2}(\uu))$ belongs to one of $(\epsilon +\mu i,\epsilon +\mu i,\epsilon +\mu i,\epsilon +\mu i)$,
$(\epsilon +\mu i,\epsilon +\mu i,-\epsilon -\mu i,-\epsilon -\mu i)$, $(\epsilon +\mu i,-\epsilon -\mu i,\epsilon -\mu i,-\epsilon +\mu i)$,
$(\epsilon +\mu i,-\epsilon -\mu i,-\epsilon +\mu i,\epsilon -\mu i)$, for $\epsilon,\mu\in\{\pm 1\}$.
  \end{enumerate}
\end{theorem}
\begin{proof}
By Lemma~\ref{lem:zs}, the generalized Walsh-Hadamard transform of $f$ (labeling $\zeta:=\zeta_{16}$) can be written as
\allowdisplaybreaks[3]
\begin{align*}
\cH_f^{(16)}(\uu )&=\sum_{\xx\in\BBF_2^n} \zeta^{b_1(\xx)+2^2 b_2(\xx)} (-1)^{\uu\cdot\xx} = \sum_{\xx\in\BBF_2^n} \zeta^{b_1(\xx)} i^{b_2(\xx)} (-1)^{\uu\cdot\xx}\\
&= \frac{1}{4} \sum_{\xx\in\BBF_2^n} \left(i^{b_1(\xx)+1} \left(-(-1)^{b_1(\xx)} \zeta^3+(-1)^{b_1(\xx)} \zeta+\zeta^3-\zeta\right)\right.\\
&\left.+
   \left(\zeta^2+1\right) \left(\left(1-(-1)^{b_1(\xx)}\right)\zeta + 1+(-1)^{b_1(\xx)}\right)+\right.\\
&\left. i^{b_1(\xx)} \left(-(-1)^{b_1(\xx)} \zeta^2+(-1)^{b_1(\xx)}-\zeta^2+1\right)
   \right) i^{b_2(\xx)} (-1)^{\uu\cdot\xx}\\
   &= \alpha \sum_{\xx\in\BBF_2^n}  (-1)^{b_1(\xx)}i^{b_1(\xx)} i^{b_2(\xx)} (-1)^{\uu\cdot\xx}+\beta \sum_{\xx\in\BBF_2^n}   i^{b_1(\xx)} i^{b_2(\xx)} (-1)^{\uu\cdot\xx}\\
   &\quad +\gamma\sum_{\xx\in\BBF_2^n}  (-1)^{b_1(\xx)}  i^{b_2(\xx)} (-1)^{\uu\cdot\xx}+\delta\sum_{\xx\in\BBF_2^n}  i^{b_2(\xx)} (-1)^{\uu\cdot\xx}\\
&=\alpha \cH_{3b_1+b_2}^{(4)} (\uu )+\beta \cH_{b_1+b_2}^{(4)} (\uu )+\gamma \cH_{2b_1+b_2}^{(4)} (\uu )+\delta\cH_{b_2}^{(4)}(\uu ),
\end{align*}
where
{\small
\allowdisplaybreaks[3]
\begin{equation}
\begin{split}
\label{E5}
8\alpha&=2(1+i\zeta-\zeta^2-i\zeta^3)= \left(2-\sqrt{2}+\sqrt{4-2 \sqrt{2}}\right)-i
   \left( \sqrt{2}-\sqrt{4-2\sqrt{2}}\right) \\
8\beta&=2(1-i\zeta-\zeta^2+i\zeta^3)= \left(2-\sqrt{2}-\sqrt{4-2 \sqrt{2}}\right)-i
   \left( \sqrt{2}+\sqrt{4-2\sqrt{2}}\right)\\
8\gamma&=2(1-\zeta+\zeta^2-\zeta^3)= \left(2+\sqrt{2}-\sqrt{4+\sqrt{2}}\right)+i
   \left( \sqrt{2}-\sqrt{4+2\sqrt{2}}\right) \\
8\delta&=2(1+\zeta+\zeta^2+\zeta^3)=  \left(2+\sqrt{2}+\sqrt{ 4+2\sqrt{2} }\right)+i
   \left( \sqrt{2}+\sqrt{4+2\sqrt{2}}\right).
   \end{split}
\end{equation}
}

First we assume that $b_2,b_1+b_2,2b_1+b_2,3b_1+b_2$ are gbent in $\cGB_n^4$, with their generalized Walsh-Hadamard transforms satisfying (i) if $n$ is even
respectively (ii) if $n$ is odd. With Equation \eqref{E5} we get the claim (to ease with the computation, we used a Mathematica program).

Conversely, we assume that $f$ is gbent in $\cGB_n^{16}$. With the same notations for the Walsh-Hadamard transforms, $A,B,C,D,W,X,Y,Z$ as in the proof of Theorem~\ref{thm-gengb16}, by Lemma~\ref{lem-Hf}, we have
\begin{align*}
2\cH_{3b_1+b_2}^{(4)}(\uu )&=(X+W)+i(X-W)\\
2\cH_{b_1+b_2}^{(4)}(\uu )&=(C+Z)+i(C-Z)\\
2\cH_{2b_1+b_2}^{(4)}(\uu )&=(B+Y)+i(B-Y)\\
2\cH_{b_2}^{(4)}(\uu )&=(A+D)+i(A-D).
\end{align*}
From Theorem~\ref{thm-gengb16}, we already know that $A,B,C,\ldots$ are two or three valued (depending upon the parity of $n$).
Running a short script through the possible values described in Theorem~\ref{thm-gengb16}, we see that the last claim of our current theorem is true as well.
\end{proof}

With the same approach as in Theorem \ref{thm-gengb16} we can also obtain results on the semibentness of functions in $\mathcal{GB}_n^{16}$.
\begin{theorem}
\label{semibent}
Let $f\in \cGB_n^{16}$ be given as
$f(\xx)=a_1(\xx)+2a_2(\xx)+2^2 a_3(\xx)+2^3a_4(\xx)$,
$a_i\in\cB_n$, $1\leq i\leq 4$.
Then $f$ is gsemibent when $n$ is odd, and generalized $2$-plateaued when $n$ is even, if and only if the Boolean function $c_1a_1\+c_2a_2\+c_3a_3\+a_4$ is semibent for all
$c_i\in\F_2$, $i=1,2,3$, such that for all $\uu\in\V_n$ their Walsh-Hadamard transforms are either all zero, or they satisfy
\begin{align*}
& \mathcal{W}_{a_4}(\uu)\mathcal{W}_{a_2+a_4}(\uu) =
\mathcal{W}_{a_3+a_4}(\uu)\mathcal{W}_{a_2+a_3+a_4}(\uu)\\
&= \mathcal{W}_{a_1+a_4}(\uu)\mathcal{W}_{a_1+a_2+a_4}(\uu)
=\mathcal{W}_{a_1+a_3+a_4}(\uu)\mathcal{W}_{a_1+a_2+a_3+a_4}(\uu), \mbox{ and}\\
& \mathcal{W}_{a_4}(\uu)\mathcal{W}_{a_3+a_4}(\uu) = \mathcal{W}_{a_1+a_4}(\uu)\mathcal{W}_{a_1+a_3+a_4}(\uu).
\end{align*}
\end{theorem}
\begin{proof}
Assume that $f$ is gsemibent in $\cGB_n^{16}$ when $n$ is odd, respectively generalized $2$-plateaued when $n$ is even. Then $|\cH_f^{(16)}(\uu)|\in\{0,\pm 2^{(n+1)/2}\}$ for $n$ odd,
respectively, $|\cH_f^{(16)}(\uu)|\in\{0,\pm 2^{(n+2)/2}\}$ for $n$ even.
Using the notations of Theorem~\ref{thm-gengb16}, from Equation~\eqref{forsystem}, we immediately get $A=B=C=D=X=Y=W=Z=0$ if $\cH_f^{(16)}(\uu )=0$.
If $|\cH_f^{(16)}(\uu )|=2^{(n+1)/2}$ (for $n$ odd), respectively, $|\cH_f^{(16)}(\uu )|=2^{(n+2)/2}$ (for $n$ even), then ~\eqref{forsystem} again
yields the system of equations ~\eqref{eq:2ndsystem} with the one difference that in the first equation the power of $2$ on the right side is $2^{n+4}$, respectively, $2^{n+5}$.
With the same argument as in the proof of Theorem~\ref{thm-gengb16}
we see that for such $\uu$, $2^{-\frac{n+1}{2}} (A, C, D, W, B, X, Y, Z)$ (for $n$ odd), respectively, $2^{-\frac{n+2}{2}} (A, C, D, W, B, X, Y, Z)$ (for $n$ even) can only take the values from
Equation~\eqref{eq-valuesH}. \\
Straightforward one confirms that the converse is also true, and the theorem is shown.
\end{proof}

\begin{remark}
Again, if some to the Boolean functions $a_i$, $i=1,2,3,4$, are the zero function, then the conditions in Theorem~\textup{\ref{semibent}} further simplify.
For instance, we can immediately see that $f(\xx) = a_1(\xx)+2^3a_4(\xx)\in\mathcal{GB}_n^{16}$ is gsemibent if $n$ is odd, respectively, generalized $2$-plateaued if $n$ is even,
if and only $a_4,a_1\+a_4$ are both semibent with $|\cW_{a_4}(\uu )|=|\cW_{a_1\+a_4}(\uu )|$, for all $\uu\in\mathbb{F}_2^n$. \\
\end{remark}

\section{Gbents in $\cGB_n^8,\cGB_n^{16}$ and their Gray image}
\label{sec4}

It was shown in \cite{smgs} that $f\in\mathcal{GB}_n^4$, with $f(\xx)=a_1+2 a_2(\xx)$, $a_1,a_2\in\mathcal{B}_n$, is gbent if and only if the Gray image $\psi(f)$ is bent
if $n$ is odd, or semibent and the associated $a_2$ and $a_1\+a_2$ have complementary autocorrelation if $n$ is even. It is the purpose of this section to extend this result.
We show that the Gray image of every gbent function in $\mathcal{GB}_n^8$ is semibent, and the Gray image of every gbent function in $\mathcal{GB}_n^{16}$ is semibent if
$n$ is odd, and $3$-plateaued if $n$ is even. We start with a lemma.
\begin{lemma}
 \label{lem-Wsum}
 Let $n,k\geq 2$ be positive integers and $F:\V_{n+k-1}\to\F_2$ be defined by $F(\xx,y_1,y_2,\ldots,y_{k-1})=a_k(\xx)\+\bigoplus_{i=1}^{k-1} y_i a_i(\xx)$, where $a_i\in\cB_n$, $1\leq i\leq k$. Denote by $\aa(\xx)$
 the vectorial Boolean function $\aa(\xx)=(a_1(\xx),\ldots, a_{k-1}(\xx))$ and let $\uu\in\V_n$ and $\vv=(v_1,\ldots,v_{k-1})\in\V_{k-1}$. The Walsh-Hadamard transform of $F$ at
$(\uu,\vv)$ is then
 \[ \cW_F(\uu ,v_1,\ldots,v_{k-1})=\sum_{\alpha\in\V_{k-1}}(-1)^{\alpha\cdot \vv}\mathcal{W}_{a_k\+\alpha\cdot\aa}(\uu). \]
 \end{lemma}
 \begin{proof}
 We will show our claim by induction on $k$.  For $k=2$ we have
 \begin{align*}
 \cW_F(\uu ,v_1)&=\sum_{\substack{\xx\in\V_n\\ y_1\in\F_2}} (-1)^{y_1 a_1(\xx)\+a_2(\xx)} (-1)^{v_1y_1\+\uu\cdot\xx}\\
 &=\sum_{\xx\in\V_n} (-1)^{a_2(\xx)} (-1)^{\uu\cdot\xx}+\sum_{\xx\in\V_n} (-1)^{a_1(\xx)\+a_2(\xx)} (-1)^{v_1\+\uu\cdot\xx}\\
 &= \cW_{a_2}(\uu )+(-1)^{v_1} \cW_{a_1\+a_2}(\uu ).
 \end{align*}
Now let
 \begin{align*}
 &F(\xx,y_1,\ldots,y_{k})=F_1(\xx,y_1,\ldots,y_{k-1})\+ y_k a_k(\xx), \text{ where }\\
 &F_1(\xx,y_1,\ldots,y_{k-1})=a_{k+1}(\xx) \+ \bigoplus_{i=1}^{k-1} y_i a_i(\xx).
 \end{align*}
 Then
 \begin{align*}
 \cW_F(\uu ,\vv,v_k)
 &= \cW_{F_1}(\uu ,\vv)+(-1)^{v_k} \cW_{F_1\+a_{k+1}}(\uu ,\vv),
 \end{align*}
 which implies our claim by the induction assumption.
 \end{proof}

\begin{theorem}
Suppose $f\in{\mathcal GB}_n^{8}$, with $f(\xx)=a_1(\xx)+2a_2(\xx)+2^2a_3(\xx)$ for all $\xx\in\V_n$, $a_1,a_2,a_3\in\cB_n$. If $f$ is gbent then $\psi(f)$ is semibent in $\cB_{n+2}$.
\end{theorem}
\begin{proof}
Recall that $\displaystyle \psi(f)(\xx,y_1, y_2)=a_1(\xx) y_1+a_2(\xx)y_2+a_3(\xx).$  By Lemma~\ref{lem-Wsum},
 \begin{equation}
 \label{eq-GB8}
 \begin{split}
\cW_{\psi(f)}(\uu ,v_1,v_2)= & \cW_{a_3}(\uu ) +(-1)^{v_1} \cW_{a_3\+a_1}(\uu )\\
& +(-1)^{v_2} \cW_{a_3\+a_2}(\uu )+(-1)^{v_1+v_2} \cW_{a_3\+a_2\+a_1}(\uu ).
 \end{split}
\end{equation}
Assume first that $n$ is even.  Since $f$ is gbent, by~\cite[Theorem 19]{smgs},
$a_3,a_1\+a_3,a_2\+a_3,a_1\+a_2\+a_3$ are all bent and $\cW_{a_3}(\uu ) \cW_{a_1\+a_2\+a_3}(\uu )=\cW_{a_1\+a_3}(\uu ) \cW_{a_2\+a_3}(\uu )$, for all $\uu\in\V_n$.
Take $\cW_{a_3}(\uu )=\mu_1(\uu ) 2^{n/2},  \cW_{a_3\+a_1}(\uu )=\mu_2(\uu ) 2^{n/2}, \cW_{a_3\+a_2}(\uu )=\mu_3(\uu ) 2^{n/2},  \cW_{a_3\+a_2\+a_1}(\uu )=\mu_4(\uu ) 2^{n/2}$,
for some $\mu_i\in\{-1,1\},1\leq i\leq 4$. Thus, $\mu_1(\uu )\mu_4(\uu )=\mu_2(\uu )\mu_3(\uu )$.
Using these in Equation~\eqref{eq-GB8}, we obtain
\[
 2^{-n/2}\cW_{\psi(f)}(\uu ,v_1,v_2)= \mu_1(\uu )+(-1)^{v_1} \mu_2(\uu )+(-1)^{v_2} \mu_3(\uu )+(-1)^{v_1\+v_2} \mu_4(\uu ).
\]
For $(\mu_1(\uu ),\mu_2(\uu ),\mu_3(\uu ),\mu_4(\uu ))$ with values in the set
\begin{align*}
&(-1,-1,-1,-1), (1, 1,-1,-1), (-1,-1, 1, 1), (-1, 1,-1, 1), \\
&(1,-1,-1, 1), (-1, 1, 1,-1), (1,-1, 1,-1), (1, 1, 1, 1),
 \end{align*}
$2^{-n/2}\cW_{\psi(f)}(\uu,v_1,v_2)$ takes one of the following values
\begin{align*}
&(-1)^{v_1\+v_2\+1}+(-1)^{v_1\+1}+(-1)^{v_2\+1}-1,\\
&(-1)^{v_1\+v_2}+(-1)^{v_1\+1}+(-1)^{v_2}-1,\\
&(-1)^{v_1\+v_2}+(-1)^{v_1}+(-1)^{v_2\+1}-1,\\
&(-1)^{v_1\+v_2\+1}+(-1)^{v_1}+(-1)^{v_2}-1,\\
& (-1)^{v_1\+v_2}+(-1)^{v_1\+1}+(-1)^{v_2\+1}+1,\\
&(-1)^{v_1\+v_2\+1}+(-1)^{v_1\+1}+(-1)^{v_2}+1,\\
&(-1)^{v_1\+v_2\+1}+(-1)^{v_1}+(-1)^{v_2\+1}+1,\\
&(-1)^{v_1\+v_2}+(-1)^{v_1}+(-1)^{v_2}+1.
\end{align*}
Therefore, $\cW_{\psi(f)}$ attains the values $0,\pm 2^{(n+4)/2}$, thus $\psi(f)$ is semibent.

We now consider the case of odd $n$. Then, by~\cite[Theorem 19]{smgs}, $a_3,a_1\+a_3,a_2\+a_3,a_1\+a_2\+a_3$ are all semibent and,
$\cW_{a_3}(\uu )=\cW_{a_1\+a_3}(\uu )=0$ and $|\cW_{a_1\+a_2\+a_3}(\uu )|=| \cW_{a_2\+a_3}(\uu )|=2^{(n+1)/2}$,
or $|\cW_{a_3}(\uu )|=|\cW_{a_1\+a_3}(\uu )|=2^{(n+1)/2}$ and $\cW_{a_1\+a_2\+a_3}(\uu )= \cW_{a_2\+a_3}(\uu )=0$.

\noindent
{\em Case $1$}. Let $\cW_{a_3}(\uu )=\cW_{a_1\+a_3}(\uu )=0$, $\cW_{a_1\+a_2\+a_3}(\uu )=\epsilon_1(\uu ) 2^{(n+1)/2}$,  $\cW_{a_2\+a_3}(\uu )=\epsilon_2(\uu )2^{(n+1)/2}$, with  $\epsilon_1,\epsilon_2$ taking values from $\{-1,1\}$.
Then from~\eqref{eq-GB8} we obtain
\[
\cW_{\psi(f)}(\uu ,v_1,v_2)=(-1)^{v_2} 2^{(n+1)/2}\left(  \epsilon_1(\uu )  +(-1)^{v_1} \epsilon_2(\uu ) \right),
\]
from which we infer that $\cW_{\psi(f)}(\uu ,v_1,v_2)\in\{0,\pm 2^{(n+3)/2}\}$, for all combinations of
$\epsilon_i(\uu )$ and $v_i$, $i=1,2$. Therefore $\psi(f)$ is semibent.

\noindent
{\em Case $2$}. Let $\cW_{a_3}(\uu )=\epsilon_1(\uu ) 2^{(n+1)/2}$, $\cW_{a_1\+a_3}(\uu )=\epsilon_2(\uu ) 2^{(n+1)/2}$, $\cW_{a_1\+a_2\+a_3}(\uu )=\cW_{a_2\+a_3}(\uu )=0$,
with  $\epsilon_1,\epsilon_2$ taking values from $\{-1,1\}$.  As before, from~\eqref{eq-GB8} we obtain
\[
\cW_{\psi(f)}(\uu ,v_1,v_2)=  2^{(n+1)/2}\left(  \epsilon_1(\uu )  +(-1)^{v_1} \epsilon_2(\uu ) \right),
\]
from which we infer that $\cW_{\psi(f)}(\uu ,v_1,v_2)\in\{0,\pm 2^{(n+3)/2}\}$ and therefore $\psi(f)$ is semibent.
\end{proof}

%
%
%
%
%
%
%

One could ask whether the converse of the previous theorem holds. In general, if we make no other assumptions on the semibent $\psi(f)$, a simple computation reveals that the possible
values for $2^{-n}|\cH_f(\uu )|^2$ (if $n$ is even), respectively, $2^{-(n-1)}|\cH_f(\uu )|^2$ (if $n$ is odd) are $4\pm 2\sqrt{2},3\pm 2\sqrt{2},2\pm \sqrt{2},3,2,1,0$ ($\cH_f(\uu)$ in both cases belongs to an 81 element set whose normalized square norms belong to the previous set of values).

\begin{sloppypar}
Also we would like to recall at this position that by Theorem 19 in \cite{smgs}, $f(\xx) = a_1(\xx)=2a_2(\xx)+2^2a_3(\xx)$ is gbent if and only if
$2^{-\frac{n}{2}}\left(\cW_{a_3}(\uu ),  \cW_{a_3\+a_1}(\uu ), \cW_{a_3\+a_2}(\uu ), \cW_{a_3\+a_2\+a_1}(\uu )\right)$ is one of the following
tuples $(-1, -1, -1, -1)$, $(-1, 1, -1, 1)$, $(-1, -1, 1, 1)$, $(-1, 1, 1, -1)$, $(1, -1, -1, 1)$, $(1, 1, -1, -1)$, $(1, -1, 1, -1)$, $(1, 1, 1, 1)$ if $n$ is even,
and $2^{-\frac{n+1}{2}}\left(\cW_{a_3}(\uu ),  \cW_{a_3\+a_1}(\uu ), \cW_{a_3\+a_2}(\uu ), \cW_{a_3\+a_2\+a_1}(\uu )\right)$ is one of the following tuples
$(-1, -1, 0, 0)$, $(0, 0, -1, -1)$, $(-1, 1, 0, 0)$, $(0, 0, -1, 1)$, $(0, 0,  1, -1)$, $(1, -1, 0, 0)$, $(0, 0, 1, 1)$, $(1, 1, 0, 0)$ if $n$ is odd.
\end{sloppypar}

In general, for a semibent function $F:\V_{n+2}\rightarrow \F_2$ of the form $F(\xx,y_1,y_2) = a_3(\xx)\+ y_1a_1(\xx)\+ y_2a_2(\xx)$, the Boolean functions $a_1,a_2,a_3$
may not satisfy those conditions. In fact, as the following example shows, a generalized Boolean function $f$ which is not gbent, may have a semibent Gray image.
Let $n=3, k=3$, and $f(x_1,x_2,x_3)=x_1+2 x_2+4 x_3$, that is, $a_1(x_1,x_2,x_3)=x_1$, $a_2(x_1,x_2,x_3)=x_2$, $a_3(x_1,x_2,x_3)=x_3$, and so,
$\psi(f)(x_1,x_2,x_3,y_1,y_2)=x_1 y_1\+x_2 y_2\+x_3$. Then the Walsh-Hadamard spectrum of $\psi(f)$ is $\{0,\pm 8\}$, and so, it is semibent, but of course, $f$ is not gbent
(since it would require $a_3, a_1\+a_3, a_2\+a_3, a_1\+a_2\+a_3$  to at least be semibent, which, certainly, they are not).

Below we present the corresponding result on the Gray image of a gbent function in $\cGB_n^{16}$.
\begin{theorem}
If $f=a_1(\xx)+2a_2(\xx)+2^2 a_3(\xx)+2^3 a_4(\xx)\in\cGB_n^{16}$ is gbent, then its Gray image $\psi(f)$ is semibent in $\cB_{n+3}$ if $n$ is odd, and $3$-plateaued in $\cB_{n+3}$
if $n$ is even.
\end{theorem}
\begin{proof}
Recall that the Gray image of $f$ is $\psi(f)(\xx,y_1,y_2,y_3)=y_1a_1(\xx)\+y_2a_2(\xx)\+y_3a_3(\xx)\+  a_4(\xx)$. By Lemma~\ref{lem-Wsum},
its Walsh-Hadamard transform is given by
\begin{align*}
&\cW_{\psi(f)}(\uu ,v_1,v_2,v_3)= \cW_{a_4}(\uu )+(-1)^{v_1} \cW_{a_4\+a_1}(\uu )\\
&\qquad+(-1)^{v_2} \cW_{a_4\+a_2}(\uu )+(-1)^{v_3} \cW_{a_4\+a_3}(\uu )\\
&\qquad
+(-1)^{v_1\+v_2} \cW_{a_4\+a_2\+a_1}(\uu )+(-1)^{v_1\+v_3} \cW_{a_4\+a_3\+a_1}(\uu )\\
&\qquad +(-1)^{v_2\+v_3} \cW_{a_4\+a_3\+a_2}(\uu )
+(-1)^{v_1\+v_2\+v_3} \cW_{a_4\+a_3\+a_2\+a_1}(\uu ).
\end{align*}
By going through the $32$ cases of Theorem~\ref{thm-gengb16} for the Walsh-Hadamard transforms in the expression above
($16$ for $n$ even and $16$ for $n$ odd), we obtain that the Walsh-Hadamard spectrum is $\{0, \pm 2^{3+n/2}\}$ (for $n$ even) and  $\{0, \pm 2^{2+(n+1)/2}\}$ (for $n$ odd), hence the claim.
\end{proof}

As we also observed for the $k=3$ case, the converse of the above theorem is not true, in general. For example, for $n=4,k=4$, let $f(x_1,x_2,x_3,x_4)=x_1+2x_2+4\cdot 1+8(x_3\+x_4)$, and so,
$a_1(x_1,x_2,x_3,x_4)=x_1$, $a_2(x_1,x_2,x_3,x_4)=x_2$, $a_3(x_1,x_2,x_3,x_4)=1$, $a_4(x_1,x_2,x_3,x_4)=x_3\+x_4$. One can see that $f$ is not gbent since the conditions of Theorem~\ref{thm-gengb16}
are not satisfied, however, the Gray image $\psi(f)$ has Walsh-Hadamard spectrum $\{0,\pm 32\}$ and so it is 3-plateaued in $\cB_7$.

Finally we observe that for $k=2,3,4$ the Gray image of a gbent function in $\mathcal{GB}_n^{2^k}$ is $(k-2)$-plateaued when $n$ is odd, and $(k-3)$-plateaued when $n$ is even, which points
to a more general theorem.

At last we want to mention that it may be worthwhile to investigate the properties of bent, semibent and plateaued functions which appear as Gray image of a gbent function.

\section{Gbent functions and bent functions}
\label{sec3.0}

In this section we present connections between gbent functions and their components for the general case of gbent functions in $\mathcal{GB}_n^{2^k}$, $k>1$.
In~\cite{ST09} it was shown that a function $f(\xx) = a_1(\xx)+2a_2(\xx)\in\mathcal{GB}_n^4$, $n$ even, is gbent if and only if $a_2$ and $a_1\+a_2$
are bent. By Theorem~19 in \cite{smgs}, for a gbent function $f\in\mathcal{GB}_n^8$, $n$ even, given as $f(\xx) = a_1(\xx)+2a_2(\xx)+4a_3(\xx)$, all
Boolean functions, $a_3$, $a_1\+a_3$, $a_2\+a_3$ and $a_1\+a_2\+a_3$ are bent. For gbent functions in $\mathcal{GB}_n^{16}$, the analog statement follows from
our Theorem \ref{thm-gengb16}. The general case is dealt with in the following theorem.
\begin{theorem}
\label{gbebe}
Let $n$ be even, and let $f(\xx)$ be a gbent function in $\mathcal{GB}_n^{2^k}$, $k>1$, (uniquely) given as
\[ f(\xx) = a_1(\xx) + 2a_2(\xx) + \cdots + 2^{k-2}a_{k-1}(\xx) + 2^{k-1}a_k(\xx), \]
$a_i\in\mathcal{B}_n$, $1\le i\le k$. Then all Boolean functions of the form
\[ g_{\mathbf{c}}(\xx) = c_1a_1(\xx) \+ c_2a_2(\xx) \+ \cdots \+ c_{k-1}a_{k-1}(\xx) \+ a_k(\xx), \]
$\mathbf{c} = (c_1,c_2,\ldots,c_{k-1}) \in \F_2^{n-1}$, are bent functions.
\end{theorem}
\begin{proof}
As in Proposition \ref{valdis}, for the gbent function $f$ we denote by $f_{\uu}$ the function $f_{\uu}(\xx) = a_1(\xx) + \cdots + 2^{k-2}a_{k-1}(\xx) + 2^{k-1}(a_k(\xx)+\uu\cdot\xx)$
in $\mathcal{GB}_n^{2^k}$. Again, the integer $b_r^{(\uu)}$, $0\le r\le 2^k-1$, is defined as $b_r^{(\uu)} = |\{\xx\in\V_n\;:\;f_{\uu}(\xx) = r\}|$. By Proposition \ref{valdis},
$b_{r+2^{k-1}}^{(\uu)} = b_r^{(\uu)}$ for all $0\le r\le 2^{k-1}-1$, except for one element $\rho_{\uu}\in\{0,\ldots,2^{k-1}-1\}$ depending on $\uu$, for which
$b_{\rho_{\uu}+2^{k-1}}^{(\uu)} = b_{\rho_{\uu}}^{(\uu)} \pm 2^{n/2}$.

Since it is somewhat easier to follow, we first show the bentness of $a_k(\xx) = g_{\00}(\xx)$. In the second step we show the general case. 
For $r \ne \rho_{\uu}$, $0\le r\le 2^{k-1}-1$, consider all $\xx\in\V_n$ for which $a_1(\xx) + \cdots + 2^{k-2}a_{k-1}(\xx) = r$.
Since $b_{r+2^{k-1}}^{(\uu)} = b_r^{(\uu)}$, for exactly half of these $\xx$ we have $a_k(\xx)+\uu\cdot\xx = 0$ (note that the number of these $\xx$ must be even).
Among all $\xx\in\V_n$ for which $a_1(\xx) + \cdots + 2^{k-2}a_{k-1}(\xx) = \rho_{\uu}$, there are $b_{\rho_u}^{(\uu)}$ for which $a_k(\xx) + \uu\cdot\xx = 0$, and
there are $b_{\rho_{\uu}+2^{k-1}}^{(\uu)} = b_{\rho_u}^{(\uu)}\pm 2^{n/2}$ for which $a_k(\xx) + \uu\cdot\xx= 1$. Hence for the Walsh-Hadamard transform of $a_k$ we get
\[ \mathcal{W}_{a_k}(\uu) = \sum_{\xx\in\V_n}(-1)^{a_k(\xx) \+ \uu\cdot\xx} = \pm 2^{n/2}, \]
which shows that $a_k$ is bent. 

To show that $g_{\mathbf{c}}$ is bent for every $\mathbf{c}\in\F_2^{k-1}$, we write $f_{\uu}(\xx)$, $\uu \in \V_n$, as
\begin{align*}
f_{\uu}(\xx) &= c_1a_1(\xx) + \cdots + c_{k-1}2^{k-2}a_{k-1}(\xx)\;+\;\bar{c}_1a_1(\xx) + \cdots + \bar{c}_{k-1}2^{k-2}a_{k-1}(\xx) \\
& + 2^{k-1}(a_k(\xx)+\uu\cdot\xx) := h(\xx) + \bar{h}(\xx) + 2^{k-1}(a_k(\xx)+\uu\cdot\xx),
\end{align*}
where $\bar{c} = c\+1$. Note that every $0\le r\le 2^{k-1}-1$ in the value set of $a_1(x) + \cdots + 2^{k-2}a_{k-2}(\xx)$ has then a unique representation as $h(\xx) + \bar{h}(\xx)$.
Consider $\xx$ for which $h(\xx) + \bar{h}(\xx) = r+s \ne \rho_{\uu}$. Again from $b_{r+2^{k-1}}^{(\uu)} = b_r^{(\uu)}$ we infer that for half of those $\xx$ we have
$a_k(\xx)\+\uu\cdot\xx = 0$. As a consequence, we also have
\[
g_{\mathbf{c}}(\xx)\+\uu\cdot\xx = c_1a_1(\xx) \+ \cdots \+ c_{k-1}a_{k-1}(\xx) \+ a_k(\xx)\+\uu\cdot\xx = 0
\]
for exactly half of those $\xx$. (Observe that $h(\xx_1) = h(\xx_2) = r$ implies $c_1a_1(\xx_1)\+\cdots\+c_{k-1}a_{k-1}(\xx_1) = c_1a_1(\xx_2)\+\cdots\+c_{k-1}a_{k-1}(\xx_2)$.)
Similarly as above, among all $\xx\in\V_n$ for which $h(\xx) + \bar{h}(\xx) = \rho_{\uu}$, there are $b_{\rho_u}^{(\uu)}$ for which $a_k(\xx) \+ \uu\cdot\xx = 0$, and
there are $b_{\rho_{\uu}+2^{k-1}}^{(\uu)} = b_{\rho_u}^{(\uu)}\pm 2^{n/2}$ for which $a_k(\xx) \+ \uu\cdot\xx= 1$. From this we conclude that
$|\{\xx\in\V_n \;:\;h(\xx) + \bar{h}(\xx) = \rho_u\,\mbox{and}\,f_{\uu}(\xx) = 1\}| - |\{\xx\in\V_n \;:\;h(\xx) + \bar{h}(\xx) = \rho_u\,\mbox{and}\,f_{\uu}(\xx) = 0\}| = \pm 2^{n/2}$.
Therefore
\[ \mathcal{W}_{g_{\mathbf{c}}}(\uu) = \sum_{\xx\in\V_n}(-1)^{g_{\mathbf{c}}(\xx) + \uu\cdot\xx} = \pm 2^{n/2}, \]
and $g_{\mathbf{c}}$ is bent.
\end{proof}

Theorem~\ref{gbebe}, which assigns to a gbent function a family of bent functions, provides a necessary condition for a function $f\in\mathcal{GB}_n^{2^k}$ to be gbent.
For $k>2$ the condition is not necessary. As the following example shows, the additional conditions on the Walsh spectra for $k=3$ given in \cite[Theorem 19]{smgs} and
for $k=4$ given in our Theorem \ref{thm-gengb16}, are required (and not implied implicitly by the bentness of the associated Boolean functions).
We remark that, as also our Theorem \ref{thm-gengb16} indicates, these additional conditions become complicated as $k$ increases.

\begin{example}
Let $n$ be even, $a_4$ be a bent function, $a_1$ be an arbitrary Boolean function, both in $\cB_n$, set $a_2:=\bar a_1,a_3:=0$. Certainly, for every triple $(c_1,c_2,c_3)$, the function
$c_1 a_1(\xx)\+c_2 a_2(\xx)\+c_3 a_3(\xx)\+a_4(\xx)=(c_1\+c_2)\+a_4(\xx)$ is bent, but the conditions of Theorem~\textup{\ref{thm-gengb16}} are not satisfied and so,
$a_1(\xx)+2 a_2(\xx)+2^2 a_3(\xx)+2^3a_4(\xx)$ is not gbent.
\end{example}

We close this section with a result which also reveals an inductive approach to the study of gbent functions in $\cGB_n^{2^k}$.

\begin{theorem}
\label{k,k-1Thm}
Let $f\in\mathcal{GB}_n^{2^k}$ with $f(\xx)=g(\xx)+2h(\xx), g\in\cB_n,h\in\cGB_n^{2^{k-1}}$.
If $n$ is even, then the following statements are equivalent.
\begin{itemize}
\item[$(i)$] $f$ is gbent in $\cGB_n^{2^{k}}$;
\item[$(ii)$] $h$ and $h+2^{k-2}g$ are both gbent in $\cGB_n^{2^{k-1}}$ with $\Im\left(\overline{\cH_h^{(2^{k-1})}(\uu )} \cH_{h+2^{k-2}g}^{(2^{k-1})}(\uu )\right)=0$,
for all $\uu\in \V_n$.
\end{itemize}
If $n$ is odd, then $(ii)$ implies $(i)$.
\end{theorem}
\begin{proof}
We first show that for $n$ even, $h$ and $h+2^{k-2}g$ are gbent in $\cGB_n^{2^{k-1}}$ if $f$ is gbent in $\cGB_n^{2^{k}}$.
In a second step, we show that if $h$ and $h+2^{k-2}g$ are both gbent in $\cGB_n^{2^{k-1}}$, then $f$ is gbent in
$\cGB_n^{2^{k}}$ if and only if $\Im\left(\overline{\cH_h^{(2^{k-1})}(\uu )} \cH_{h+2^{k-2}g}^{(2^{k-1})}(\uu )\right)=0$,
for all $\uu\in \V_n$. This will conclude the proof for both, $n$ even and $n$ odd. 

Let $\uu\in\V_n$, and for $e\in\{0,1\}$ and $r\in\{0,\ldots,2^{k-1}-1\}$, let
\[ S^{(\uu)}(e,r) = \{\xx\in\V_n\;:\;g(\xx) = e\;\mbox{and}\;h(\xx)+2^{k-2}(\uu\cdot\xx) = r\}. \]
With the notations of Proposition \ref{valdis}, we have $f_{\uu}(\xx) = f(\xx) + 2^{k-1}(\uu\cdot\xx) = g(\xx)+2(h(\xx)+2^{k-2}(\uu\cdot\xx))$,
and $|S^{(\uu)}(e,r)| = b^{(\uu)}_{e+2r}$.
If $f$ is gbent, by Proposition \ref{valdis}, there exist $\epsilon\in\{0,1\}$ and $0\le \rho_{\uu}\le 2^{k-2}-1$, for which
$|S^{(\uu)}(\epsilon,\rho_{\uu}+2^{k-2})| = |S^{(\uu)}(\epsilon,\rho_{\uu})| \pm 2^{n/2}$.
For $(e,r) \ne (\epsilon,\rho_{\uu})$, we have $|S^{(\uu)}(e,r+2^{k-2})| = |S^{(\uu)}(e,r)|$.
Observing that $\{\xx\in\V_n\;:\;h(\xx)+2^{k-2}(\uu\cdot\xx) = r\} = S^{(u)}(0,r)\cup S^{(u)}(1,r)$, we obtain
\[ \mathcal{H}_h^{(2^{k-1})}(\uu) = \sum_{\xx\in\V_n}\zeta_{2^{k-1}}^{h(\xx)}(-1)^{\uu\cdot\xx} = \sum_{\xx\in\V_n}\zeta_{2^{k-1}}^{h(\xx)+2^{k-2}(\uu\cdot\xx)} = \pm\zeta_{2^{k-1}}^{\rho_{\uu}}2^{n/2}. \]
Consequently, $h$ is gbent in $\mathcal{GB}_n^{2^{k-1}}$. For $h+2^{k-2}g\in\mathcal{GB}_n^{2^{k-1}}$ we have
\begin{align*}
\mathcal{H}_{h+2^{k-2}g}^{(2^{k-1})}(\uu) &= \sum_{\xx\in\V_n}\zeta_{2^{k-1}}^{h(\xx)+2^{k-2}(\uu\cdot\xx) + 2^{k-2}g{\xx}} =
\sum_{e\in\F_2\atop r\in\Z_{2^{k-1}}}\sum_{x\in S^{(u)}(e,r)}\zeta_{2^{k-1}}^{r + 2^{k-2}e} \\
& = \sum_{e\in\F_2\atop r\in\Z_{2^{k-1}}} |S^{(u)}(e,r)|\zeta_{2^{k-1}}^{r + 2^{k-2}e} = \pm\zeta_{2^{k-1}}^{\rho_{\uu} + 2^{k-2}\epsilon}2^{n/2},
\end{align*}
which implies that also $h+2^{k-2}g$ is gbent in $\mathcal{GB}_n^{2^{k-1}}$. \\[.3em]
To show the condition on $\Im\left(\overline{\cH_h^{(2^{k-1})}(\uu )} \cH_{h+2^{k-2}g}^{(2^{k-1})}(\uu )\right)$, we write $\zeta_{2^k}=x+yi$, $\cH_h^{(2^{k-1})}(\uu )=a+b i$
and $\cH_{h+2^{k-2}g}^{(2^{k-1})}(\uu )=c+di$. From Equation \eqref{eq:gb2k}, taking the complex norm, squaring and rearranging terms (recall that $|\zeta_{2^k}|^2=x^2+y^2=1$),
we get
\begin{align*}
2|\cH_f^{(2^k)}(\uu )|^2 &=\frac{1}{2}(a^2 + b^2) (1 + 2 x + x^2 + y^2) +\frac{1}{2} (c^2 + d^2) (1 - 2 x + x^2 +  y^2)\\
    &\qquad  \qquad -  (a c + b d) (\xx^2 + y^2 - 1) + 2 (ad-b c) y\\
&=  |\cH_h^{(2^{k-1})}(\uu )|^2 (1 + x) + |\cH_{h+2^{k-2}g}^{(2^{k-1})}(\uu )|^2 (1 -  x)\\
  &\quad
 + 2 y\, \Im\left(\overline{\cH_h^{(2^{k-1})}(\uu )} \cH_{h+2^{k-2}g}^{(2^{k-1})}(\uu )\right).
\end{align*}
If $h$, $h+2^{k-2}g$ are gbent, i.e. $|\cH_h^{(2^{k-1})}(\uu )|^2 = |\cH_{h+2^{k-2}g}^{(2^{k-1})}(\uu )|^2 = 2^n$ for all $\uu\in\V_n$,
then we immediately see that $|\cH_f^{(2^k)}(\uu )|^2 = 2^n$ for all $\uu\in\V_n$, and hence $f$ is gbent if and only if
$\Im\left(\overline{\cH_h^{(2^{k-1})}(\uu )} \cH_{h+2^{k-2}g}^{(2^{k-1})}(\uu)\right)=0$, for all $\uu\in\V_n$.
\end{proof}

\begin{remark}
For $n$  even and $k=2$, Theorem~\textup{\ref{k,k-1Thm}} recovers the result in~\textup{\cite{ST09}} on the relation between gbentness and bentness of the components
(see also~\textup{\cite[Corollary 15 \& 16]{smgs}}): The function $f(\xx) = a_1(\xx)+2a_2(\xx) \in \mathcal{GB}_n^4$
is gbent if and only if both Boolean functions $a_2$ and $a_1\+a_2$ are bent. \\
If $n$ is odd, as an example for the implication $(ii)\Longrightarrow (i)$ we can take $g=0$, and an arbitrary $h$ gbent in  $\cGB_n^{2^{k-1}}$. Certainly, the conditions from $(ii)$ are readily satisfied.
\end{remark}

\begin{remark}
In~\textup{\cite{HP}}, conditions are derived for the gbentness of some functions $f\in \mathcal{GB}_n^q$ of the form $f(\xx) = \frac{q}{2}a(x) + rb(x)$,
$r\in[q/4,3q/4]$, $a,b$ in $\mathcal{GB}_n^q$ or $\mathcal{B}_n$.
\end{remark}

Analyzing the components of gbent functions in $\mathcal{GB}_n^{2^k}$, $n$ even, we obtained some necessary conditions on gbentness (Theorem~\ref{gbebe}) and necessary and sufficient conditions
on gbentness (Theorem~\ref{k,k-1Thm}), where the latter are sufficient but not necessary also for odd $n$, however not very simple to check.
Presumably, one could also attempt to extend our Theorem~\ref{thm-gengb16} to the case $k=5$, although, we doubt that the complicated equations one would obtain can easily be analyzed,
and certainly they do not give any further insight into the nature of gbent functions.
We believe that the next step in completely characterizing gbentness for all $k$, should be to find a more ``inductive'' approach,
where one would connect gbentness of $f$ in $\mathcal{GB}_n^{2^k}$ to its components in $\mathcal{GB}_n^{2^{k-j}}$, $1\leq j\leq 4$, and using the results in this paper. \\[.3em]

\noindent
{\bf Acknowledgements.} Work by P.S. started during a very enjoyable visit at RICAM.
Both the second and third named author thank the institution for the excellent working conditions. 
The second author is supported by the Austrian Science Fund (FWF) Project no. M 1767-N26.

\end{document}